\documentclass[a4paper,UKenglish,cleveref,autoref,numberwithinsect]{lipics-v2019}

\usepackage[defaultlines=4,all]{nowidow}

\title{Parameterized Distributed Complexity Theory:  \protect\\ A logical approach}

\titlerunning{Parameterized Distributed Complexity Theory}

\author{Sebastian Siebertz}{University of Bremen, Germany}{siebertz@uni-bremen.de}{https://orcid.org/0000-0002-6347-1198}{}

\author{Alexandre Vigny}{University of Bremen, Germany}{vigny@uni-bremen.de}{https://orcid.org/0000-0002-4298-8876}{}

\authorrunning{S.\ Siebertz, A.\ Vigny}

\Copyright{Sebastian Siebertz, Alexandre Vigny}

\ccsdesc[500]{Theory of computation~Models of computation}
\ccsdesc[500]{Theory of computation~Graph algorithms analysis}
\ccsdesc[500]{Theory of computation~Parameterized complexity and exact algorithms}

\keywords{Local, Congest,Congested-Clique, Distributed
Computing, Distributed Parameterized Complexity}

\category{}

\nolinenumbers 

\hideLIPIcs  

\EventEditors{John Q. Open and Joan R. Access}
\EventNoEds{2}
\EventLongTitle{42nd Conference on Very Important Topics (CVIT 2016)}
\EventShortTitle{CVIT 2016}
\EventAcronym{CVIT}
\EventYear{2016}
\EventDate{December 24--27, 2016}
\EventLocation{Little Whinging, United Kingdom}
\EventLogo{}
\SeriesVolume{42}
\ArticleNo{1}

\usepackage{comment}
\usepackage{relsize,xspace}
 \usepackage{xcolor}
 \usepackage{todonotes}
 \usepackage{comment}
\usepackage{microtype}
\usepackage{anyfontsize}
\usepackage{bm}
\usepackage{tikz}
\usepackage{refcount}
\usepackage{wrapfig}
\usepackage[arrow,matrix,curve]{xy}
\usepackage{marginnote}
\usepackage[single]{accents}
\usepackage{enumerate}

\newcommand{\class}[1]{\mathcal #1}

\newcommand{\Ff}{\mathcal{F}}

\definecolor{blue}{rgb}{0.2,0.2,0.9}
\definecolor{brown}{rgb}{0.6,0.6,0.2}

\newcommand{\Oof}{\mathcal{O}}
\newcommand{\CCC}{\class{C}}

\newcommand{\PPP}{\textsc{P}}

\newcommand{\GGG}{\class{G}}

\newcommand{\wh}[1]{\widehat{#1}}

\newcommand{\coloneqq}{:=}

\newcommand{\N}{\mathbb{N}}

\renewcommand{\phi}{\varphi}
\renewcommand{\epsilon}{\varepsilon}

\newcommand{\FO}{\mathrm{FO}}
\newcommand{\dist}{\mathrm{dist}}

\renewcommand{\a}{\overline{a}}
\renewcommand{\b}{\overline{b}}
\newcommand{\ov}{\overline{v}}
\newcommand{\x}{\overline{x}}
\newcommand{\y}{\overline{y}}
\newcommand{\z}{\overline{z}}

\newcommand{\wLambda}{\widehat{\Lambda}}

\newcommand{\wphi}{\widehat{\varphi}}

\newcommand{\parent}{\mathsf{parent}}

\renewcommand{\leq}{\leqslant}
\renewcommand{\geq}{\geqslant}
\renewcommand{\le}{\leqslant}

\newcommand{\norm}[1]{\left\lVert#1\right\rVert}

\newcommand{\Distributed}{\ensuremath{\mathsf{DIS\-TRI\-BU\-TED}}\xspace}
\newcommand{\Local}{\ensuremath{\mathsf{LOCAL}}\xspace}
\newcommand{\Congest}{\ensuremath{\mathsf{CONGEST}}\xspace}
\newcommand{\CCongest}{\ensuremath{\mathsf{CONGESTED\text{-}CLIQUE}}\xspace}
\newcommand{\LFPT}{\ensuremath{\mathsf{LOCAL}}\text{-}\ensuremath{\mathsf{FPT}}\xspace}

\newcommand{\LocalW}[1]{\ensuremath{\mathsf{LOCAL}}\text{-}\ensuremath{\mathsf{W}}[#1]\xspace}

\newcommand{\LocalA}[1]{\ensuremath{\mathsf{LOCAL}}\text{-}\ensuremath{\mathsf{A}}[#1]\xspace}

\newcommand{\LocalAW}{\ensuremath{\mathsf{LOCAL}}\text{-}\ensuremath{\mathsf{AW}[\star]}\xspace}

\newcommand{\FPT}{\ensuremath{\mathsf{FPT}}\xspace}

\newcommand{\Weft}{\ensuremath{\mathsf{WEFT}}\xspace}
\newcommand{\Pp}{\mathcal{P}}

\crefformat{page}{#2page~#1#3}%
\Crefformat{page}{#2Page~#1#3}%
\crefformat{equation}{#2(#1)#3}%
\Crefformat{equation}{#2(#1)#3}%
\crefformat{figure}{#2Figure~#1#3}%
\Crefformat{figure}{#2Figure~#1#3}%
\crefformat{section}{#2Section~#1#3}
\Crefformat{section}{#2Section~#1#3}
\crefformat{chapter}{#2Chapter~#1#3}
\Crefformat{chapter}{#2Chapter~#1#3}
\crefformat{chapter*}{#2Chapter~#1#3}
\Crefformat{chapter*}{#2Chapter~#1#3}
\crefformat{part}{#2Part~#1#3}
\Crefformat{part}{#2Part~#1#3}
\crefformat{enumi}{#2(#1)#3}
\Crefformat{enumi}{#2(#1)#3}

\usepackage{enumerate}

\usepackage{latexsym}

\hypersetup{
    colorlinks = true,
	linkcolor={blue},
	citecolor={blue}
}

\begin{document}
\maketitle


\begin{abstract}
\noindent Parameterized complexity theory offers a framework
for a refined analysis of hard algorithmic problems. Instead of
expressing the running time of an algorithm as a function
of the input size only, running times are expressed with respect
to one or more parameters of the input instances.
In this work we follow the approach of
parameterized complexity to provide a framework
of \emph{parameterized distributed complexity}.
The central notion of efficiency in parameterized
complexity is fixed-parameter tractability and we define the
distributed analogue $\Distributed$-$\textsf{FPT}$ (for
$\Distributed\in \{\Local, \Congest, \CCongest\}$) as the class of
problems that can be solved in $f(k)$ communication
rounds in the \Distributed model of distributed computing,
\linebreak where~$k$ is the parameter of the problem instance and $f$ is an
arbitrary computable function. To classify hardness we
introduce three hierarchies. The
$\Distributed$-\textsf{WEFT}-hierarchy
is defined analogously to the
\textsf{W}-hierarchy in parameterized
complexity theory via reductions to the weighted circuit
satisfiability problem, but it turns out that this definition
does not lead to
satisfying frameworks for the \Local and \Congest models.
We then follow a logical approach that leads to a more
robust theory. We define the levels of the
$\Distributed$-\textsf{W}-hierarchy and the
$\Distributed$-\textsf{A}-hierarchy that have first-order
model-checking problems as their complete problems via
suitable reductions.
\end{abstract}

\section{Introduction}

The \emph{synchronous message passing model}, which can be traced back
at least to the seminal paper of Gallager, Humblet and
Spira~\cite{gallager1983distributed},
is a theoretical model of distributed systems that allows to focus on
certain important aspects of distributed computing. In this model, a
distributed system is modeled by an undirected (connected) graph~$G$,
in which each vertex $v\in V(G)$ represents a computational entity of
the network, often referred to as a node of the network, and each edge
$\{u,v\}\in E(G)$ represents a bidirectional communication channel
that connects the two nodes $u$ and~$v$.  The nodes are equipped with
unique numerical identifiers (of size $\Oof(\log n)$, where $n$ is the
order of the network graph). In a distributed algorithm, initially,
the nodes have no knowledge about the network graph and only know
their own and their neighbors' identifiers.  The nodes then
communicate and coordinate their actions by passing messages to one
another in order to achieve a common goal.

The synchronous message
passing model without any bandwidth restrictions is called the
\Local~model of distributed computing~\cite{peleg2000distributed}.  If
every node is restricted to send messages of size at most~$\Oof(\log n)$ one obtains the \Congest~model, and finally, if
messages of size~$\Oof(\log n)$ can be sent to all nodes of the
network graph (not only to the neighbors of a node) we speak of the
\CCongest~model.  The time complexity of a distributed algorithm in
each of these models is defined as the number of communication rounds
until all nodes terminate their computations.


Typically considered computational tasks are related to graphs, in
fact, often the graph that describes the network topology is the graph
of the problem instance itself.~For example, in a distributed
algorithm for the \textsc{Dominating Set} problem,
the~\mbox{computational} task is to compute a small dominating set of
the network graph $G$. Each node of the network must decide and report
whether it shall belong to the dominating set or not.


Research in the distributed computing community is to a large extent
problem-driven. There is a huge body of literature on upper and lower
bounds for concrete problems.  We refer to the surveys of
Suomela~\cite{suomela2013survey} and Elkin~\cite{elkin2004distributed}
for extensive overviews of distributed algorithms. There has also been
major progress in developing a systematic distributed complexity
theory, including definitions of suitable locality preserving
reductions and distributed complexity classes. We refer
to~\cite{BalliuBOS18,BalliuHKLOS18,ChangP17,FeuilloleyF16,
  fraigniaud2013towards, KorhonenS18} for extensive background.

A very successful approach to deal with computationally hard problems
is the approach of parameterized
complexity. Instead of measuring the running time of an algorithm with
respect to the input size only, this approach takes into
account one or more additional parameters. In many practical applications it is reasonable to assume
that structural parameters of the input instances are bounded, another
commonly considered parameter is the size of the solution. In case a
parameter is bounded, one can design special algorithms that aim to
restrict the non-polynomial dependence of the running time to this
parameter. For example, the currently fastest known exact algorithm
for the \textsc{Dominating Set} problem on $n$-vertex graphs runs in time
$\Oof(1.4969^n)$~\cite{van2011exact}. If, however, we are dealing with
structured graphs, e.g.\ if we may assume that a graph~$G$ excludes a
complete bipartite subgraph $K_{t,t}$, we can decide in time
$2^{\Oof(t^2k\log k)}\cdot \norm{G}$ whether~$G$ contains a dominating
set of size at most $k$~\cite{fabianski2018progressive}. When~$k$ and
$t$ are small and $G$ is large, this may be a major improvement over
the exact algorithm.  If a problem admits such running times, we speak
of a fixed-parameter tractable problem. More precisely, a
parameterized problem is \emph{fixed-parameter tractable} if there is
an algorithm solving it in time $f(k)\cdot n^c$, where~$k$ is the
parameter, $f$ is a computable function, $n$ is the input size and $c$
is a constant.

\smallskip In this work we follow the approach of parameterized
complexity to provide a framework of \emph{parameterized distributed
  complexity}.
For any \Distributed model, where $\Distributed \in \{\Local,
\Congest, \mathsf{CONGESTED}$-$\mathsf{CLIQUE}\}$, we
define the distributed complexity class $\Distributed$-$\textsf{FPT}$ as the class of problems that can be
solved in $f(k)$ communication rounds in the \Distributed model, where
$k$ is the parameter of the problem instance and~$f$ is an arbitrary
computable function.  These classes are the distributed analogues of
the central notion of fixed-parameter tractability.



\smallskip Parameterized approaches to distributed computing were
recently studied in~\cite{KorhonenR17}, where it was shown that
$k$-paths and trees on $k$ nodes can be detected in $\Oof(k\cdot
2^k)$ rounds in the \textsf{BROADCAST}-\textsf{CONGEST} model. Similar
randomized algorithms were obtained in the context of distributed
property testing~\cite{Censor-HillelFS16, EvenFFGLMMOORT17}. The
setting which is closest to our present work is the work of Ben-Basat
et al.~\cite{ben2018parameterized}.  The authors studied the
parameterized distributed complexity of several fundamental graph
problems, parameterized by solution size, such as the \textsc{Vertex
  Cover} problem, the \textsc{Independent Set} problem, the
\textsc{Dominating Set} problem, the \textsc{Matching} problem, and
several more. In each of these problems the question is to decide
whether there exists a solution of size $k$, where $k$ is the input
parameter. They showed that all of the above problems are
fixed-parameter tractable in the \Local~model -- in our notation: they
belong to the class $\Local$-$\textsf{FPT}$.

This is no surprise e.g.\ for the \textsc{Dominating Set} problem: if
a dominating set of a connected graph $G$ has size at most~$k$,
then the diameter of $G$ is bounded by~$3k$. Hence, in the
\Local~model one can either learn in~$3k$ rounds the whole graph
topology and determine by brute force whether a dominating set of
size $k$ exists. Otherwise, if the diameter is too large, the
algorithm can simply reject the instance as a negative
instance. Similarly,  an independent set of size $k$ can be chosen
greedily if the diameter of $G$ is sufficiently large and the problem
can be solved by brute force otherwise. The authors
of~\cite{ben2018parameterized} formalized this phenomenon by
defining the class $\textsf{DLB}$ of problems whose optimal solution
size is lower bounded by the graph diameter.
The situation is more complex in the \Congest~model. For this model,
the authors study two problems, namely the \textsc{Vertex Cover}
problem and the \textsc{Matching} problem, and prove that both
problems admit fixed-parameter distributed algorithms in the
\Congest~model -- in our notation: they belong to the class
\Congest-$\textsf{FPT}$.

\smallskip In parameterized complexity theory, the \textsc{Vertex
  Cover} problem is a standard example of a fixed-parameter tractable
problem, while the \textsc{Independent Set} problem and
\textsc{Dominating Set} problem are beliebved to be intractable. While
lacking the techniques to actually prove this intractability,
parameterized complexity theory offers a way to establish
intractability by classifying problems into complexity classes by
means of suitable reductions. The \textsf{W}-hierarchy is a collection
of complexity classes that may be seen as a parameterized refinement
of the classical complexity class \textsf{NP}.  The
\textsc{Independent Set} problem is the foremost example of a problem
that is hard for the parameterized complexity class
$\textsf{W}[1]$. Similarly, the \textsc{Dominating Set} problem is a
prime example of a $\textsf{W}[2]$-hard problem.  The
\textsf{A}-hierarchy is a collection of complexity classes that may be
seen as a parameterized analogue of the polynomial hierarchy.

\smallskip As the \textsc{Independent Set} problem and the
\textsc{Dominating Set} problem are in \LFPT, these problems cannot
take the exemplary role of hard problems that they take in classical
parameterized complexity theory.  However, their colored variants
remain hard also in the distributed setting.  For example, in the
\textsc{Multicolored Independent Set} problem one searches in a
colored graph for an independent set where all vertices of the set
have different colors.
%
%
As the colors can be given to vertices in an arbitrary way, the
problem looses its local character and becomes hard also in the \Local
model.  In the \Congest model this hardness is already observed for
the uncolored \textsc{Independent Set}
problem~\cite{ben2018parameterized}. The authors of~\cite{ben2018parameterized} establish a lower
bound of $\Omega(n^2/\log^2 n)$ on the number of rounds in the
\Congest model, where $n$ can be arbitrarily larger than $k$.

\smallskip In classical parameterized complexity theory, the
$\mathsf{W}$-hierar\-chy is defined by the complexity of circuits that
is required to check a solution. This hierarchy was introduced by
Downey and Fellows in~\cite{downey1995fixed}.  At a first glance it
seems natural to model the circuit evaluation problem as a graph
problem and consider it as such in the distributed setting. This leads
to the definition of a problem class $\textsc{Weft}[t]$ for each
$t\geq 1$ and we define the class $\Distributed$-$\mathsf{WEFT}[t]$ as
the class of those problems that reduce via parameterized \Distributed
reductions to a member of $\textsc{Weft}[t]$.  The problem with this
is that in the definition of the $\mathsf{W}$-hierarchy one considers
circuit families of bounded depth. In our locality sensitive setting
this does not lead to a robust complexity theory: for example, we
obtain that $\Local$-$\mathsf{WEFT}[t]\subseteq \LFPT$ for each~$t$,
while the \textsc{Multicolored Independent Set} problem, which we
would like to place into the class \Local-$\mathsf{WEFT}[1]$, does not
lie in any of the classes $\Local$-$\mathsf{WEFT}[t]$.  These problems
do not arise in the \CCongest~model and the \CCongest-\textsf{WEFT}
hierarchy is an interesting hierarchy to study.  We refer to
\Cref{sec:weft} for the details.

\smallskip
The \textsf{A}-hierarchy was introduced by Flum and
Grohe, originally in terms of the parameterized halting problem for
alternating Turing machines~\cite{flum2001fixed}. This definition
cannot be easily adapted to the distributed setting.  Instead, we
follow another approach of~\cite{flum2001fixed} (see also the
monograph~\cite{flum2006parameterized}), where problems are classified
by their descriptive complexity.  More precisely, the authors classify
problems by the syntactic form of their definitions in first-order
predicate logic. This leads in a very natural way to the definition of
the levels of the $\mathsf{W}$- and $\mathsf{A}$-hierarchy.  We denote
by $\Sigma_0$ and $\Pi_0$ the class of quantifier-free formulas. For
$t\geq 0$, we let~$\Sigma_{t+1}$ be the class of all formulas
$\exists x_1\ldots\exists x_\ell\,\phi$, where $\phi\in \Pi_t$, and we
let $\Pi_{t+1}$ be the class of all formulas
$\forall x_1\ldots\forall x_\ell\,\phi,$ where $\phi\in \Sigma_t$.
Furthermore, for $t\geq 1$, $\Sigma_{t,1}$ denotes the class of all
formulas of~$\Sigma_t$ such that all quantifier blocks after the
leading existential block have length at most~$1$.  The model-checking
problem for a class~$\Phi$ of formulas, denoted $\textsc{MC}$-$\Phi$,
is the problem to decide for a given (vertex and edge colored) graph
$G$ and formula $\phi\in \Phi$ whether~$\phi$ is satisfied on $G$.
For all $t\geq 1$, $\textsc{MC}$-$\Sigma_{t,1}$ is complete
for~$\textsf{W}[t]$ under fpt-reductions, and for all $t\geq 1$,
$\textsc{MC}$-$\Sigma_t$ is complete for $\textsf{A}[t]$ under
fpt-reductions~\cite{flum2006parameterized}. After giving an
appropriate notion of parameterized \Distributed reductions, we define
for all $t\geq 1$ the class \Distributed-$\textsf{W}[t]$ as the class
$[\textsc{MC}$-$\Sigma_{t,1}]^{\Distributed}$ of problems that reduce
to the~$\Sigma_{t,1}$ model-checking problem via parameterized
\Distributed reductions. Analogously, we define for all $t\geq 1$ the
class \Distributed-$\textsf{A}[t]$ as the class
$[\textsc{MC}$-$\Sigma_{t}]^{\Distributed}$ of problems that reduce to
the $\Sigma_{t}$ model-checking problem via parameterized \Distributed
reductions. The details are presented in \Cref{sec:hierarchy}.


Let us comment on our choice to use the model-checking problem for
first-order logic as the basis of our distributed complexity
theory. In principle, one could take any problem and define a
complexity class from its closure under appropriate reductions. The
model-checking problem for fragments of first-order logic is a very
natural candidate to use for the definition of complexity classes. The
number of quantifiers and of quantifier alternations in a formula
needed to describe a problem give an intuitive indication about the
complexity of the problem, which naturally leads to a hierarchy of
complexity classes.
First-order logic can express many important graph problems in an
elegant way, e.g.\ the existence of a multicolored independent set of
size~$k$ can be expressed by the formula
$\exists x_1\ldots \exists x_k (\bigwedge_{1\leq i\leq
  k}P_i(x_i)\wedge \bigwedge_{1\leq i\neq j\leq k} \neg
E(x_i,x_j)\big)$.
Here, the $P_i$ are unary predicates that encode the colors of
vertices and $E$ is a binary predicate that encodes the edge relation
(here we assume for simplicity that the graph is colored only with the
colors $P_1,\ldots, P_k$).  The above formula is a
$\Sigma_{1,1}$-formula, hence the \textsc{Multicolored Independent
  Set} problem is placed in the class $\Distributed$-$\textsf{W}[1]$,
as intended.  Similarly, the existence of a dominating set of size at
most $k$ can be expressed by the formula
$\exists x_1\ldots\exists x_k \forall y \big(\bigvee_{1\leq i\leq k}
(y=x_i \vee E(y,x_i))\big)$.
This is a $\Sigma_{2,1}$-formula, which places the \textsc{Dominating
  Set} problem in the class $\Distributed$-$\textsf{W}[2]$.  In
particular, we have the desired inclusions
$\Distributed$-$\textsf{FPT}\subseteq \Distributed$-$\mathsf{W}[1]$
and
$\Distributed$-$\mathsf{WEFT}[t]\subseteq
\Distributed$-$\textsf{W}[t]$
for all $t\geq 1$. The details can be found in \Cref{sec:hierarchy}.

We then use a classical theorem from model theory, namely Gaifman's
Theorem, to prove that the \Local-\textsf{W}- and
\Local-\textsf{A}-hierarchy collapse to the second level of the
\Local-\textsf{W}-hierarchy.
On the other
hand we conjecture that the
\Congest-\textsf{W}- and~-\textsf{A}- and the \CCongest-\textsf{W}-
and -\textsf{A}-hierarchies are strict.

Since $\Sigma_{1,1}=\Sigma_1$, we have
\Local-\textsf{W}$[1]=$ \Local-\textsf{A}$[1]$ (just as
in classical parameterized complexity theory).
In the light
of the above collapse result it remains (at least for first-order
definable problems) to determine whether they belong to
\Local-\textsf{W}$[1]$. We
prove that the $\textsc{Multicolored Independent
Set}$ problem and the $\textsc{Induced Subgraph Isomorphism}$ problem
are complete for \Local-\textsf{W}$[1]$ under \Local reductions.
In classical parameterized complexity theory these problems
are prime examples of \textsf{W}$[1]$-complete problems,
however, the standard reductions do not translate to \Local
reductions, and we have to come up with new reductions. We prove that
$\LocalW{1} \subsetneq \LocalW{2}$ by showing that the 
\textsc{Clique Domination} problem is not in $\LocalW{1}$.
The details can be found in
\cref{seq:first-levels}.



 We then turn our attention to \emph{distributed
   kernelization}.  Kernelization is a classical approach in
 parameterized complexity theory to reduce the size of the input
 instance in a polynomial time preprocessing step. More formally, a
 kernelization for a parameterized problem $\PPP$ is an algorithm that
 computes for a given instance $(G,k)$ of $\PPP$ in time polynomial in
 $(|G|+k)$ an instance $(G',k')$ of $\PPP$ such that $(G,k)$ is a
 positive instance of~$\PPP$ if and only if $(G',k')$ is a positive
 instance of $\PPP$ and such that $(|G'|+k')$ is bounded by a
 computable function in $k$.  The output $(G',k')$ is called a
 kernel. It is a classical result of parameterized complexity that a
 problem is fixed-parameter tractable if and only if it admits a
 kernel.  We give two definitions of distributed kernelization
 and study their relationship to fixed-parameter tractability.
 %
 %
 The details are given in \Cref{sec:kernel}.

 \smallskip

 Finally, we define the class $\Distributed$-$\mathsf{XPL}$ as the
 class of problems that can be solved in $f(k)\cdot (\log n)^{g(k)}$
 rounds (for computable functions $f$ and $g$) in the \Distributed
 model. $\mathsf{XPL}$ stands for \emph{slicewise poly-logarithmic}. In
 parameterized complexity theory the class~$\mathsf{XP}$ of slicewise
 polynomial problems contains all problems that can be solved in
 time~$n^{g(k)}$ for some computable function $g$. This definition
 obviously has to be adapted to make sense in the distributed setting,
 as every problem can be solved in a polynomial number of rounds
 (polynomial in the graph size) in the \Congest~model. As the final
 result we show that the model-checking problem of first-order logic is
 in \CCongest-\textsf{XPL} when parameterized by formula length on
 classes of graphs of bounded expansion. We conjecture that this is not
 the case on all graphs.  The details are presented in
 \Cref{sec:mc-be}.

\medskip
\section[Distributed-FPT and reductions]{Distributed fixed-parameter tractability and
reductions}\label{sec:fpt-red}

\medskip
We consider the \emph{synchronous message passing model},
in which
a distributed system is modeled by an undirected connected
graph~$G$. Each vertex $v\in V(G)$ represents a computational
entity of the network, often referred to as a node of the network,
and each edge $\{u,v\}\in E(G)$ represents a bidirectional
communication channel that connects the two nodes~$u$ and~$v$.
The nodes are equipped with unique
numerical identifiers (of size $\Oof(\log n)$, where $n$
is the order of the network graph). In a distributed algorithm,
initially, the nodes have no knowledge
about the network graph and only know their own and their
neighbors identifiers.
The nodes communicate and coordinate their actions by
passing messages to one another in order to achieve a common goal.
The synchronous message
passing model without any bandwidth restrictions is called the
\Local~model of distributed computing~\cite{peleg2000distributed}.
If every node is
restricted to send messages of size at most $\Oof(\log n)$
one obtains the \Congest~model, and finally, if messages
of size $\Oof(\log n)$ can be sent to all nodes of the
network graph (not only to neighbors) we speak of the \CCongest~model.
The time complexity of a distributed
algorithm in each of these models is defined as the number
of communication rounds until all nodes terminate their
computations. We refer to the surveys~\cite{elkin2004distributed,
FeuilloleyF16, suomela2013survey} for
extensive overviews of distributed algorithms.

Typically considered computational tasks
are related to graphs, in fact, often the graph that
describes the network topology is the graph of the problem
instance itself. We therefore focus on graph
problems, and, as usual in complexity theory and also
parameterized complexity theory, on decision problems.
We allow a fixed number
of vertex and edge labels/colors that are accessible via unary
and binary predicates $P_1,\ldots, P_s\subseteq V(G)$ and
$E_1,\ldots, E_t\subseteq V(G)^2$, for fixed $s,t\in \N$.
We write $\GGG_{s,t}$
for the set of all
finite connected graphs with $s$ unary and~$t$ binary
predicates. In the following, an instance
of a decision problem is a pair $(G,k)$, where~$G$ is a connected,
vertex and edge colored
graph, and $k\in \N$ is a parameter.
We refer to the textbooks \cite{CyganFKLMPPS15, DowneyF13,flum2006parameterized}
for extensive background on parameterized complexity theory.

\begin{definition}
A parameterized decision problem is a set $\PPP\subseteq
\GGG_{s,t}\times\N$ for some $s,t\in \N$.
\end{definition}

Often in the literature, see e.g.~\cite{fraigniaud2013towards},
in a distributed algorithm for a decision problem, each
processor must produce a Boolean output \emph{accept}
or \emph{reject} and the
decision is defined as the conjunction of all outputs.
As a matter of taste, we instead define the decision of an algorithm
as the \emph{disjunction} of all outputs. With this definition,
for example the decision problem whether a graph contains
a vertex of a fixed degree~$d$ becomes a local problem.


\begin{definition}
In a distributed algorithm for a parameterized decision problem
$\PPP$, each node has access to the parameter $k$
of the instance and must at termination produce
an output \emph{accept} or \emph{reject}. The decision
of the algorithm is defined as the disjunction of
the outputs of all nodes, i.e., if the instance belongs
to $\PPP$, then some processor must accept
and otherwise, all processors must
reject.
\end{definition}

We come to the central notion of distributed
fixed-parame\-ter tractability. In the following let
$\textsf{DISTRIBUTED}$ be any of \Local, \Congest, or
\CCongest.

\begin{definition}
A parameterized decision problem $\PPP$
belongs to $\mathsf{DISTRIBUTED}$-$\mathsf{FPT}$ if there exists
a computable function $f$ and a $\mathsf{DISTRIBUTED}$ algorithm
that on input $(G,k)$ correctly decides in time $f(k)$ whether $(G,k)\in \PPP$.
\end{definition}

\smallskip
We remark that the nodes do not have to know the function
$f$, however, the algorithm must guarantee that all nodes terminate
after $f(k)$ steps. It is immediate from the definitions that
\Congest-\textsf{FPT} $\subseteq$ \Local-\textsf{FPT} and
\Congest-\textsf{FPT} $\subseteq$ \CCongest-\textsf{FPT}.

\begin{example}
\textsc{Independent Set} $\in$ \Local-\textsf{FPT} and
\textsc{Dominating Set} $\in$ \Local-\textsf{FPT}.
\end{example}
\begin{proof}
This was proved in~\cite{ben2018parameterized} and is a simple
consequence of the fact that the size of a maximum independent
set (and minimum dominating set) is functionally lower bounded by the
graph diameter.
\end{proof}

The input to
the \textsc{Multicolored Independent Set} problem is an
integer $k$ and a graph
$G\in \GGG_{k,1}$ for some $k\geq 1$ (recall that $\GGG_{s,t}$ denotes the set of
finite connected graphs with $s$ unary and $t$ binary predicates).
If a node $v$ that is labeled by a predicate $P_i$
is said to have color $P_i$. The algorithmic question is to decide
whether there exist $k$ elements with pairwise different colors
that form an independent set.

\begin{lemma}\label{lem:mis-not-fpt}
\textsc{Multicolored Independent Set} is not in $\Local$-$\mathsf{FPT}$.
\end{lemma}
\begin{proof}
Assume that there exists a $\Local$-$\mathsf{FPT}$ algorithm $\mathcal{A}$ that
works in $c\coloneqq f(3)$ rounds on each instance $G\in \mathcal{G}_{3,1}$. Consider the graph $G_n\in \mathcal{G}_{3,1}$ for $n>3c$
with vertex set $\{v_1,\ldots, v_n\}$
and edge set $\{\{v_i,v_{i+1}\}~:~1\leq i\leq n-1\}$, where
$v_1$ is colored with $P_1$,
$v_n$~is colored with $P_2$ and~$v_{n/2}$ is
colored with $P_3$. By definition,  some node $v_i$ must
accept the input. By symmetry we may assume that $1\leq i\leq n/2$.

Now consider the graph $G_n'$, which is a copy of $G_n$,
where every vertex gets the same node
identifier but only $v_1$ gets color $P_1$ and $v_{n/2}$ gets
color $P_3$, and $v_n$ does not get color~$P_2$. Then $G_n'$ is a
negative instance, however, the $c$-neighborhood of $v_i$ in $G_n'$
is equal to the $c$-neighborhood in $G_n$. Hence, $v_i$ also accepts
the instance $G_n'$, a contradiction.
\end{proof}

We come to the definition of \Distributed 
reductions.

\begin{definition}\label{def:reduction}
A \Distributed reduction is a $\mathsf{DISTRIBUTED}$ algorithm that
turns an instance $(G,k)$ of a parameterized problem
into an instance $(G',k')$, where $(G',k')$ is en\-co\-ded in
$(G,k)$ as follows.
There is a mapping \mbox{$\nu\colon V(G')\rightarrow
V(G)$} and a mapping $\eta\colon E(G')\rightarrow \mathcal{P}(G)$,
where~$\mathcal{P}(G)$ denotes all paths in $G$, with the property that
if $e=\{x,y\}\in E(G')$, then~$\eta(e)$ is a path between $\nu(x)$
and $\nu(y)$ in $G$. The mappings
are stored in vertices of $G$, more precisely, each vertex $v\in V(G)$
stores all $x\in V(G')$ such that
$\nu(x)=v$ and all paths $\eta(\{x,y\})$ such that $\{x,y\}\in E(G')$.
The \emph{radius}
of the reduction is the length of the longest path in the
image of $\eta$ and its \emph{congestion}
is the largest number of paths in the image of $\eta$
that a single edge of~$G$ belongs
to.

For computable functions $s, r,c,t,p$ we say that the reduction
is $(s,r,c,t,p)$-\emph{bounded} if the order of $G'$ is bounded
by~$|G|^{s(k)}$, the radius of the reduction
is bounded by $r(k)$, its congestion is bounded by $c(k)$,
the reduction is computable in $t(k)$ rounds and the parameter satisfies
$k'\leq p(k)$.
\end{definition}

Observe that we implicitly require that all nodes compute the
same parameter $k'$, as we turn the instance $(G,k)$ into the
instance $(G',k')$ and the parameter must be known to all
nodes.


\begin{definition}
For parameterized problems $\PPP_1$ and $\PPP_2$
we write $\PPP_1\leq_{\Distributed} \PPP_2$ if
there exist computable
functions $s, r,c,t$ and $p$ and an $(s,r,c,t,p)$-bounded
\Distributed reduction that maps any instance
$(G,k)$ to an instance $(G',k')$ such that
$(G,k)\in \PPP_1 \Leftrightarrow (G',k')\in \PPP_2$.
If $\Distributed=\Local$ we allow unbounded
congestion.
\end{definition}

\begin{definition}
Let $\Pp$ be a set of parameterized problems. We write
$[\Pp]^{\Distributed}$ for the set of all problems
$\PPP$ with $\PPP\leq_\Distributed \PPP'$ for some $\PPP'\in \Pp$.
\end{definition}

The next lemma shows that distributed reductions preserve
fixed-parameter tractability, as desired. After turning the instance
$(G,k)$ of $\PPP_1$ into an equivalent instance $(G',k')$
of $\PPP_2$, we simulate
the passing of a message from $x$ to $y$ along an edge of $G'$ by
passing it along the path $\eta(\{x,y\})$ in $G$.

\begin{lemma}\label{lem:preservation-reduction}
Let $\PPP_1\leq_\Distributed \PPP_2$ and assume that $\PPP_2$
is in \Distributed-\FPT. Then also $\PPP_1$ is in \Distributed-\FPT.
\end{lemma}
\begin{proof}
We need to show that there exists a \Distributed~algorithm
that on input
$(G,k)$ decides whether $(G,k)\in \PPP_1$ in~$g(k)$
rounds, for some computable function $g$. The algorithm
proceeds as follows. Let $n\coloneqq |V(G)|$.

As $\PPP_1\leq_\Distributed \PPP_2$, there exist
computable functions $s, r,c,t$ and $p$ and
an $(s,r,c,t,p)$-bounded \Distributed reduction that
maps any instance $(G,k)$ to an instance $(G',k')$ such that
$(G,k)\in \PPP_1 \Longleftrightarrow (G',k')\in \PPP_2$
and such that $k'\leq p(k)$. In case $\Distributed=\Local$
we may have unbounded congestion.
On input $(G,k)$ we apply this reduction and compute in $t(k)$
rounds an instance $(G',k')$ with the above properties.
We write $\nu$ and $\eta$ for the mappings representing the
graph $G'$ in $G$ (see \Cref{def:reduction}).

As $\PPP_2$
is in \Distributed-\FPT, there exists a computable function $f$
so that we can solve the instance $(G',k')$ in $f(k')\leq f(p(k))$
steps in the \Distributed~model. We simulate the run of
this algorithm on $(G',k')$ in the graph $G$. A message that is
sent by the algorithm on~$(G',k')$ may have size at most
$\Oof(\log |G'|)=\Oof(s(k)\cdot \log n)$,
unless \Distributed = \Local.
Whenever
a message is supposed to be sent along an edge $\{x,y\}\in E(G')$,
we send this message between the
appropriate vertices $u$ and $v$ of $G$ such that
$\nu(x)=u$ and $\nu(y)=v$ along the path $\eta(\{x,y\})$.
The path $\eta(\{x,y\})$ has length at most $r(k)$, it can
hence be encoded by $r(k)\cdot \log n$ bits that we send
along with the message to make routing in $G$ possible (the
factor $\log n$ comes from the ids of the vertices of
size $\log n$). Due to constraints on congestion, we may not
be able to send all messages at once. However, by assumption,
each edge of $G$ appears in at most~$c(k)$ paths
$\eta(\{x,y\})$. Hence, each message can be sent after
waiting for at most~$c(k)$ rounds until the transmission line is free.
Thus, the simulation of sending one message takes at most
$s(k)\cdot c(k)\cdot r(k)^2$ rounds.  As the functions~$r$ and
$c$ are computable, we can compute the number
$s(k)\cdot c(k)\cdot r(k)^2$ and synchronize the network accordingly.
In case $\Distributed=\Local$ we do not have to wait for
free transmission lines and can transmit each message in time
$r(k)$, and synchronize the network accordingly.

The total running time of the algorithm is hence $g(k)\coloneqq
t(k)+f(p(k))\cdot r(k)^2\cdot c(k)\cdot s(k)$, which is again a computable
function. After this time, every node $v\in V(G)$
returns the disjunction of the answers of the nodes
$x\in V(G')$ with $\nu(x)=v$. Hence, the algorithm accepts
$(G,k)$ if and only some node in~$G'$ accepts $(G',k')$. Hence,
the algorithm correctly decides~$\PPP_1$ in $g(k)$ rounds,
as desired.
\end{proof}


%

The following example is simple but instructive, as it is not a
valid parameter preserving reduction in classical parameterized
complexity.

\begin{example}\label{example:clique-ds}
\noindent\textsc{Clique Domination} $\leq_\Local$ \textsc{Red-Blue Dominating Set}.
\end{example}

\begin{proof}
In the \textsc{Clique Domination} problem we get as input a graph
$G$ and two integers $k,\ell\in \N$. The problem is to decide
whether $G$ contains a set of at most $k$ vertices that dominates
every clique of size (exactly) $\ell$ in $G$. The parameter is
$k+\ell$. The input to the \textsc{Red-Blue Dominating Set} problem
is a graph $G$ whose vertices are colored red or blue,
and an integer $k\in \N$. The problem is to decide whether
there exists a set of at most~$k$ red vertices that dominate all
blue vertices. The parameter is $k$.

On input $(G,k,\ell)$ we create a copy of each clique
of size~$\ell$ and color the vertices of the newly created vertices
blue. All original vertices are colored red. We denote an original
vertex $v$ by $(v,0)$ and its copies by $(v,i)$ for some $i\in\N$.
Introduce all edges $\{(u,i), (v,0)\}$ for $i\in \N$ and
\mbox{$\{u,v\}\in E(G)$}. The mapping $\nu:V(G')
\rightarrow V(G)$ maps each vertex $(u,i)\in V(G')$ to the
vertex $u$ in $V(G)$ and the mapping $\eta: E(G')\rightarrow
\Pp(G)$ maps each edge $\{(u,i), (v,0)\}$ to the edge
$\{u,v\}\in E(G)$.  Observe that the resulting graph $G'$ can
have size $\Oof(n^\ell)$, where $n=|V(G)|$, and that this
reduction has unbounded congestion. Clearly, the graph $G'$
contains a set of $k$ red vertices dominating all blue vertices
if and only if the graph $G$ contains a set of $k$ vertices
dominating all cliques of size $\ell$. Observe that we can compute
the reduction in a constant number of rounds, hence we can
set $t(k)=t$ for some $t$. We can set $s(k+\ell)=k+\ell$,
$r(k+\ell)=1$, $p(k+\ell)=k+\ell$. Then, the above is an
$(s,r,\infty,t,p)$-bounded \Local reduction.
\end{proof}

\begin{lemma}\label{lem:reductions-transitive}
If $\PPP_1\leq_\Distributed \PPP_2$ and $\PPP_2\leq_\Distributed \PPP_3$, then \mbox{$\PPP_1\leq_\Distributed \PPP_3$}.
\end{lemma}

\begin{proof}
As $\PPP_1\leq_\Distributed \PPP_2$, there exist
computable functions $s_1, r_1,c_1,t_1$ and~$p_1$ and
an $(s_1,r_1,c_1,t_1,p_1)$-bounded 
\Distributed reduction that
maps any instance $(G_1,k_1)$ to an instance $(G_2,k_2)$
such that \mbox{$(G_1,k_1)\in \PPP_1 \Longleftrightarrow (G_2,k_2)
\in \PPP_2$}
and such that $k_2\leq p(k_1)$. As $\PPP_2\leq_\Distributed \PPP_3$, there exist
computable functions $s_2, r_2,c_2,t_2$ and~$p_2$ and an 
$(s_2,r_2,c_2,t_2,p_2)$-bounded \Distributed reduction that
maps any instance $(G_2,k_2)$ to an instance $(G_3,k_3)$
such that 
$(G_2,k_2)\in \PPP_2 \Longleftrightarrow (G_3,k_3)
\in \PPP_3$
and such that $k_3\leq p(k_2)$.
In case $\Distributed=\Local$
we may have unbounded congestion.
We combine these reductions as in the proof of
\Cref{lem:preservation-reduction} to obtain an
$(s_3,r_3,c_3,t_3,p_3)$-bounded reduction
from $\PPP_1$ to $\PPP_3$. As $|G_2|\leq |G_1|^{s_1(k_1)}$ and
$k_2\leq p_1(k_1)$,
we have $|G_3|\leq |G_1|^{s_1(k_1)\cdot s_2(p_1(k_1))}$, and
we can define $s_3(k)\coloneqq s_1(k)\cdot s_2(p_1(k))$.
Similarly, we can define $r_3(k)\coloneqq r_2(p_1(k))\cdot r_1(k)$
and $c_3(k)\coloneqq c_2(p_1(k))\cdot c_1(k)$.
For the construction of $G_3$ in $G_2$ we need to simulate
sending a message between two vertices by routing along
an appropriate path, just as in \Cref{lem:preservation-reduction}.
Observe that we need to send node identifiers of size
$\log |G_2|$ here, which is bounded by $s_1(k_1)\cdot\log |G_1|$.
Hence, we get an additional factor $s_1(k_1)$ in the function
bounding $t_3$, that is, we can set $t_3(k)\coloneqq
t_1(k_1)+ t_2(p_1(k))\cdot r_1(k)^2\cdot c_1(k)\cdot s_1(k)$.
Finally, we can define $p_3(k)\coloneqq p_2(p_1(k))$.
\end{proof}

\section{The distributed $\textsf{WEFT}$-hierarchy}\label{sec:weft}

We now define the \Distributed $\textsf{WEFT}$-hierarchy
analogously to the classical $\textsf{W}$-hierarchy.
In these definitions we interpret the network graph as
a circuit.

\begin{definition}
A \emph{Boolean decision circuit} with $n$ inputs is a tuple
$C=(V,E,\beta)$, where $(V,E)$ is a finite directed
acyclic graph, $\beta: V\rightarrow \{\neg, \vee,
\wedge, \bigvee, \bigwedge\}\cup \{x_1,\ldots, x_n\}$,
such that the following conditions hold:
\begin{enumerate}
\item If $v\in V$ has in-degree $0$, then $\beta(v)\in
\{x_1,\ldots, x_n\}$. These vertices are the
\emph{input gates}.
\item If $v\in V$ has in-degree $1$, then
$\beta(v)=\neg$.
\item If $v\in V$ has in-degree $2$, then
$\beta(v)\in \{\vee,\wedge\}$. Vertices with degree $\leq 2$ are called \emph{small gates}.
\item If $v\in V$ has in-degree at least $3$, then
$\beta(v)\in \{\bigvee, \bigwedge\}$. These vertices are called \emph{large gates}.
\item $(V,E)$ has exactly one vertex of out-degree $0$,
called the \emph{output gate}.
\end{enumerate}
\end{definition}

The circuit computes a function $f_C:\{0,1\}^n\rightarrow \{0,1\}$
in the expected way. We refer to the textbook~\cite{vollmer2013introduction}
for more background on circuit complexity.

To encode a circuit as a
colored graph, observe that the input gates are the only vertices of
in-degree $0$. The label $\beta(v)$ for an input gate $v$ is encoded
by the unique node id by a number between $1$ and $n$
of size at most $\log n$. All other gates are assigned a color in
$\{\neg, \vee,\wedge, \bigvee, \bigwedge, o\}$, marking their
type.

%

\begin{definition}
The \emph{depth} of a circuit $C$ is defined to be the
maximum number of gates (small or large) on an input-output
path in $C$. The \emph{weft} of a circuit $C$ is the maximum
number of large gates on an input-output path in $C$.
\end{definition}

\begin{definition}
We say that a family of circuits $\Ff$ has boun\-ded depth if there is
a constant~$h$ such that every circuit in the family $\Ff$ has
depth at most $h$. We say that $\Ff$ has bounded weft if
there is a constant $t$ such that every circuit in the family $\Ff$
has weft at most $t$. 
A decision circuit $C$ accepts an
input vector $x$ if the single output gate has value~$1$ on input
$x_1,\ldots, x_n$. The \emph{weight} of a boolean vector is the number
of $1$'s in the vector.
\end{definition}

The following definition was given by Downey and Fellows~\cite{downey1995fixed}.

\begin{definition}
Let $\Ff$ be a family of decision circuits (we allow that
$\Ff$ may have many different circuits with a given number
of inputs). To $\Ff$ we associate the parameterized circuit problem
$\Pp_\Ff\coloneqq \{(C,k) : C\in \Ff$  and $C$ accepts an input vector of weight $k\}$. For $t\geq 1$, the class \textsc{Weft}$[t]$ consists of all
parameterized circuit problems $\Pp_\Ff$, where each
circuit in~$\Ff$ has depth bounded by some universal constant
and weft at most $t$.
\end{definition}

We are ready to define the \Distributed \textsf{WEFT}-hierarchy.

\begin{definition}
For $t\geq 1$ we define
$\Distributed\text{-}\Weft[t]\coloneqq [\textsc{Weft}[t]]^\Distributed.$
\end{definition}

The following are standard examples from parameterized complexity
theory, see e.g.~\cite{CyganFKLMPPS15}.
We present the proof
in the appendix for completeness.

\begin{example}\label{ex:mis-weft}
$\textsc{Multicolored Independent Set}\in 
  \CCongest$-$\textsf{WEFT}[1]$ and $\textsc{Red-Blue}$ 
  $\textsc{Dominating Set}\in \CCongest$-$\textsf{WEFT}[2]$.
\end{example}

\begin{proof}[Proof of \cref{ex:mis-weft}]
These are standard examples from parameterized complexity
theory, see e.g.~\cite{CyganFKLMPPS15}. We present
the construction for \textsc{Multicolored Independent
Set} for completeness.

On input $(G,k)$, we construct a circuit of weft $1$ and height~$3$ for
\textsc{Multicolored Independent Set}. We have one input
gate for every vertex $v$ of $G$ that we
identify with the vertex~$v$. The gates which in a satisfying
assignment are assigned the value $1$ will correspond one-to-one
with a multicolored independent set in $G$. To express this,
we state that neither two vertices of the same color, nor two
adjacent vertices can be picked into the multicolored independent
set. Hence, we connect each input with a negation gate and we
write $(\neg v)$ for the corresponding node of the circuit. Now,
for each edge $\{u,v\}\in E(G)$ and for each pair~$(u,v)$
such that $u$ and $v$ have the same color, we introduce one node
$(\neg u \vee \neg v)$ that we connect with
the nodes $(\neg u)$ and $(\neg v)$. Finally, we connect all
these disjunction gates in a big conjunction, which is the output
gate. It is easy to see that
a satisfying assignment corresponds one-to-one to a
multicolored clique.

We have to show that we can construct the circuit with a bounded
\CCongest reduction in a constant number of rounds.
The vertex map $\nu:V(C)\rightarrow V(G)$ takes every
node $v$ and $(\neg v)$ to the vertex
$v$ and every node labeled $(\neg u\vee \neg v)$ to the
smaller (referring to vertex ids in the network graph) of
$u$ and $v$. Assuming $u$ is smaller than $v$, we map the
edge from $(\neg u)$ to $(\neg u\vee \neg v)$ to the length
$0$ path $u,u$ and the edge from~$v$ to the edge~$\{u,v\}$, which is an edge in the congested clique. Finally,
we choose an arbitrary vertex~$v$ that represents the big
conjunction. The edges from this conjunction are mapped to
the vertices that represent the vertices. In total, on each
edge we have congestion at most~$3$.
\end{proof}


The above example shows that the \CCongest \Weft-hierarchy
 is an interesting hierarchy to study. In particular,
we conjecture that the \CCongest \Weft-hierarchy is strict.

Opposed to this,
the \Local \Weft-hierarchy does not behave as
intended. The following lemma is a simple consequence of the
fact that all circuits in \textsc{Weft}$[t]$ have bounded height
and one designated output gate. Hence, for each problem $\PPP_\Ff$
in \textsc{Weft}$[t]$, the radius of each circuit of $\Ff$ is
bounded by a constant, and therefore a \Local algorithm can
learn the whole circuit in a constant number of rounds and
solve the corresponding decision problem.

\begin{lemma}\label{lem:localweft}
For all $t\geq 1$, \Local-$\mathsf{WEFT}[t]\subseteq $
\Local-\FPT.
\end{lemma}

According to \Cref{lem:mis-not-fpt}, \textsc{Multicolored Independent
Set} $\not\in$ \Local-\textsf{FPT}, and hence the problem also
does not belong to \Local-\textsf{WEFT}$[t]$ for any $t\geq 1$.
This does not reflect our intuition about the complexity of the
problem, and hence, in the following section we define the
\Distributed-\textsf{W}
hierarchy in a different way.
We obviously have \Congest-\textsf{WEFT}$[t]\subseteq \Local$-\textsf{WEFT}$[t]$ for all $t\geq 1$, hence, \textsc{Multicolored Independent
Set} does not belong to \Congest-\textsf{WEFT}$[t]$ for any $t\geq 1$.
Hence, the \Congest-\textsf{WEFT} does not behave as we intend,
nevertheless, it may be an interesting hierarchy to study.


\section{The \textsf{W}- and \textsf{A}-hierarchy}\label{sec:hierarchy}

We follow the approach of Grohe and Flum~\cite{flum2001fixed, flum2006parameterized}
and define the \Distributed $\textsf{W}$- and \textsf{A}-hierarchy
via logic. \emph{First-order formulas} over a vocabulary of
vertex and edge colored graphs
$\{P_1,\ldots, P_s, E_1,\ldots, E_t\}$ are formed from
atomic formulas~$x=y$, $P_i(x)$, and $E_j(x,y)$, where each
$P_i$ is a unary relation symbol and each $E_j$ is
a binary relation symbol, and~$x,y$ are variables (we
assume that we have an infinite supply of variables) by the usual
Boolean connectives~$\neg$~(negation), $\wedge$ (conjunction),
 $\vee$ (disjunction) and existential and universal
quantification~$\exists x,\forall x$ over vertices, respectively.
The free variables of a formula are those not in the scope
of a quantifier, and we write~$\phi(x_1,\ldots,x_k)$ to indicate that
the free variables of the formula~$\phi$ are among $x_1,\ldots,x_k$. A
\emph{sentence} is a formula without free variables.

To define the semantics, we inductively define a
satisfaction relation~$\models$, where for a colored graph $G$,
a formula $\phi(x_1,\ldots,x_k)$, and elements~$a_1,\ldots,a_k\in V(G)$, 
$G\models\phi(a_1,\ldots,a_k)$
means that~$G$ satisfies~$\phi$ if the free variables~$x_1,\ldots,x_k$
are interpreted by~$a_1,\ldots,a_k$, respectively.
We refer to the textbook~\cite{libkin2013elements} for extensive
background on first-order logic over finite structures.

\begin{definition}
Both $\Sigma_0$ and $\Pi_0$ denote the class of
quanti\-fier-free formulas. For $t\geq 0$, we
let $\Sigma_{t+1}$ be the class of all formulas
$\exists x_1\ldots\exists x_\ell\,\phi$,
where $\phi\in \Pi_t$, and $\Pi_{t+1}$ the class of all formulas
$\forall x_1\ldots\forall x_\ell\,\phi$,
where $\phi\in \Sigma_t$. Furthermore, for $t\geq 1$,
$\Sigma_{t,1}$ denotes the class of all formulas of~$\Sigma_t$
such that all quantifier blocks after the leading existential
block have length at most~$1$.
\end{definition}

For every vocabulary we fix some computable encoding (injective) function
$\mathrm{enc}(\cdot)$ that takes as input a formula $\phi$ and
outputs a  bitstring (a string over $\{0,1\}$). For example,
$\mathrm{enc}(\phi)$ could be a string representation of the
syntactic tree of $\phi$.
We write $|\phi|$ for the length
of the encoding $|\mathrm{enc}(\phi)|$. We also fix some
computable bijection $\mathrm{num}(\cdot)$ from bitstrings
to $\N$.

\begin{definition}
The model-checking problem
for a set~$\Phi$ of sentences, denoted $\textsc{MC}$-$\Phi$,
is the problem to decide for a given (colored) graph $G$
and sentence $\phi\in \Phi$, whether
$\phi$ is satisfied on $G$. The parameter is $|\phi|$.
\end{definition}

By Corollary 7.27 of~\cite{flum2006parameterized},
for all $t\geq 1$, $\textsc{MC}$-$\Sigma_{t,1}$ is
complete for~$\textsf{W}[t]$ under fpt-reductions and
by Definition 5.7 and Lemma 8.10 of~\cite{flum2006parameterized}
for all $t\geq 1$, $\textsc{MC}$-$\Sigma_t$ is complete
for $\textsf{A}[t]$ under fpt-reductions. We want to take this as
the definition for the analogous distributed hierarchies, however,
we have to clarify how we represent formulas as integer parameters.

\begin{definition}
The distributed model-checking problem
for a set~$\Phi$ of sentences, denoted $\textsc{DMC}$-$\Phi$,
is the problem to decide for a given (colored) graph $G$
and integer $k$, if there is $\phi\in \Phi$ with
$\mathrm{num}(\mathrm{enc}(\phi))=k$, such that
$\phi$ is satisfied on $G$. The parameter is $k$.
\end{definition}

Observe that, since there are only a bounded number of formulas
of a fixed length, there is a computable monotone function $f$
such that $|\phi|\leq f(k)$ and $k\leq f(|\phi|)$, where
$\mathrm{num}(\mathrm{enc}(\phi))=k$. Hence, the two
definitions are fpt-equivalent, and the latter is suitable for
our framework of distributed computing. Due to the above
observation we do not distinguish between a formula and
the integer representing it.

\begin{definition}
For $t\geq 1$, we let
$\Distributed\text{-}\mathsf{A}[t]\coloneqq [\textsc{DMC}\text{-}\Sigma_t]^\Distributed$
and \linebreak
$\Distributed\text{-}\mathsf{AW}[\star]=  [\textsc{DMC}\text{-}\FO]^\Distributed.$
\newline Similarly, we let
$\Distributed\text{-}\mathsf{W}[t]\coloneqq [\textsc{DMC}\text{-}\Sigma_{t,1}]^\Distributed$
\end{definition}

\begin{example}\label{example:w2}
$\textsc{Multicolored Independent Set}\in \Distributed$-$\mathsf{W}[1]$, and \newline \textsc{Red-Blue} \textsc{Dominating Set} $\in \Distributed$-$\mathsf{W}[2]$.
\end{example}
\begin{proof}
The existence of a multicolored independent set of size~$k$ can
be expressed by the $\Sigma_{1,1}$-formula:
$~  \exists x_1\ldots \exists x_k (\bigwedge_{1\leq i\leq k}P_i(x_i)\wedge \bigwedge_{1\leq i\neq j\leq k} \neg E(x_i,x_j)\big)$.
Here, the $P_i$ are unary predicates that encode the colors of
vertices (we assume for simplicity that the graph is colored only with
the colors $P_1,\ldots, P_k$).
Similarly, the existence of a
red-blue dominating set of size at most $k$ can be expressed by
a $\Sigma_{2,1}$-formula.
\end{proof}

It is immediate from the definitions that for all $t\geq 1$ we have
\Distributed-\textsf{W}$[t]\subseteq \Distributed$-\textsf{A}$[t]\subseteq \Distributed$-\textsf{AW}$[\star]$. Furthermore,
we have \Distributed-\textsf{W}$[1]=\Distributed$-\textsf{A}$[1]$.
We conjecture that the above inclusions are strict and that we have
proper hierarchies in the \Congest and \CCongest model.

\begin{lemma}
\Distributed-$\mathsf{WEFT}[t]\subseteq \Distributed$-$\mathsf{W}[t]$.
Furthermore, if \\$\Distributed\in $ $\{\Local, \Congest\}$, then
the inclusion is strict.
\end{lemma}
\begin{proof}
The inclusion follows from the fact that satisfiability of a circuit of weft
$t$ can be expressed by a $\Sigma_{t,1}$-formula, see
Lemma 7.26 of~\cite{flum2006parameterized}. Furthermore,
According to \Cref{lem:mis-not-fpt}, \textsc{Multicolored Independent
Set} $\not\in$ \Local-\textsf{FPT} and according to
\cref{lem:localweft} the problem also
does not belong to \Local-\textsf{WEFT}$[t]$ for any $t\geq 1$.
\end{proof}

\begin{lemma}
$\Distributed$-$\FPT\subseteq \Distributed$-$\mathsf{W}[1]$.
\end{lemma}
\begin{proof}
Assume $\PPP\in \Distributed$-$\FPT$ and let $(G,k)$
be an instance of $\PPP$. By assumption there exists an
algorithm that decides in $f(k)$ rounds whether $(G,k)\in \PPP$,
that is, each node produces in $f(k)$ rounds a binary output
accept or reject. We use this algorithm as a reduction to
$\textsc{MC}$-$\Sigma_{1,1}$. We introduce a unary predicate
$P_1$ such that exactly the nodes~$v$ that accept $(G,k)$
satisfy $P_1(v)$. Then $(G,k)\in \PPP \Leftrightarrow
G\models\exists x P_1(x)$. Hence, $\PPP\in [\textsc{MC}\text{-}\Sigma_{t,1}]^\Distributed$.
\end{proof}


\section{The first levels of the LOCAL hierarchies}\label{seq:first-levels}
\subsection{Collapse of the hierarchies}

We now study the \Local hierarchies, which behave different
than expected. E.g.\ \Cref{example:clique-ds} and
\Cref{example:w2} imply that \mbox{$\textsc{Clique Domination}\in
\Local$-$\textsf{W}[2]$}. \textsc{Clique Domination} is a
classical example of an \textsf{A}$[2]$-complete problem
(in classical parameterized complexity). Observe that the
problem can be formulated as a $\Sigma_{2}$ formula,
while the fact that~$\ell$ is an input parameter makes it impossible
to express it as a $\Sigma_{2,1}$ formula. It is the infinite computational power of individual nodes in the \Local
model that makes it possible to reduce the problem to the
 \textsc{Red-Blue Dominating Set} problem with parameter~$k$ only.
Even more surprisingly,
we show that in fact the \LocalAW (and
hence the whole \Local-\textsf{A}-hierarchy) collapses
to \LocalW{2}.

We follow that by showing that the first levels of the hierarchies are not equal,
which gives the full pictures of the \Local-hierarchy.

\begin{theorem}\label{thm-local-hierarchy}
  $\LFPT \subsetneq \LocalW{1} \subsetneq \LocalW{2} = \LocalAW$
\end{theorem}

The first inclusion is strict thanks to \cref{lem:mis-not-fpt} and \cref{example:w2}.
The other statements of \cref{thm-local-hierarchy} are proved by \cref{lem:aw=w2}
and \cref{lem:w1-w2}.

\begin{lemma}\label{lem:aw=w2}
$\LocalAW =\LocalW{2}$.
\end{lemma}
\begin{proof}
By a classical theorem of Gaifman~\cite{gaifman1982local},
every sentence~$\phi$ of first-order logic is equivalent
to a computable Boolean combination of sentences of the form
\[\exists x_1\ldots \exists x_s \Big( \bigwedge_{1\leq i\leq s}
\alpha^{(r)}(x_i) \wedge \bigwedge_{1\leq i<j\leq s}
\mathrm{dist}(x_i,x_j)>2r\Big),\]
where $s\leq k+1$ if $k$ is the quantifier-rank of $\phi$,
$r\leq 7^k$, and $\alpha^{(r)}(x)$ is an $r$-local property
of~$G$, i.e., its truth depends only on the isomorphism type
of the $r$-neighborhood of the free variable $x$ in $G$.
A sentence of the above form is called a \emph{basic local
sentence}.

Now, any problem $\PPP$ in \Local-\textsf{AW}$[\star]$ reduces
via a \Local reduction to the model-checking problem for a
first-order sentence $\phi$. We translate $\phi$ into the
Boolean combination of sentences as described above.
By a local reduction we now compute a new graph $H$ and
a sentence $\psi\in \Sigma_{2,1}$ such that
$G\models\phi\Leftrightarrow H\models \psi$. We proceed by a
chain of reductions, which can be combined to the desired
reduction by \cref{lem:reductions-transitive}. We first show
how to handle negations of basic local sentences.

\begin{claim}\label{claim:negation}
For every graph $G$ and basic local sentence $\phi$
we can compute a graph $H$ and sentences $\psi_i'\in \Sigma_{2,1}$ for
$0\leq i<k$ such that
$G\models\neg\phi\Leftrightarrow H'\models\bigvee_i\psi_i'$.
\end{claim}
\begin{proof}
Assume $\phi=\exists x_1\ldots \exists x_s \Big( \bigwedge_{1\leq i\leq s}
\alpha^{(r)}(x_i) \wedge \bigwedge_{1\leq i<j\leq s}
\mathrm{dist}(x_i,x_j)>2r\Big)$. We evaluate for every vertex~$a$ of $G$ whether
$\alpha^{(r)}(a)$ holds in $G$. This is possible, as $\alpha^{(r)}$
is only a local property that can be evaluated in the \Local
model by brute force. We assign to every vertex
$a$ for which $\alpha^{(r)}(a)$ holds the color $P_\alpha$. We call
a vertex with color $P_\alpha$ and \emph{$\alpha$-vertex}. We now
compute a new set of edges, which we call \emph{$\alpha$-edges}.
We connect two $\alpha$-vertices by an $\alpha$-edge if their distance
is at most $2r$. The set of $\alpha$-edges is obviously also computable
by a local algorithm. An \emph{$\alpha$-component} is a connected component in
the graph induced by the $\alpha$-edges. We claim that there does not
exist a set of $k$ $\alpha$-vertices with pairwise distance greater than
$2r$ in $G$ if and only if
\begin{enumerate}
\item every $\alpha$-component has diameter (with respect to $\alpha$-edges)
smaller than $(2r+1)k$,
\item there exist at most $\ell<k$ $\alpha$-components $C_1,\ldots, C_\ell$, and
\item if $k_i$ denotes the size of a largest subset of $\alpha$-vertices from
$C_i$ with pairwise
distance greater than $2r$, then $\sum_{1\leq i\leq \ell}k_i<k$.
\end{enumerate}

First assume that there does not
exist a set of $k$ $\alpha$-vertices with pairwise distance greater than
$2r$ in $G$. Then
\begin{enumerate}
\item no $\alpha$-component has diameter
at least $(2r+1)k$. Otherwise, there exist two $\alpha$-vertices $a, b$ and
a shortest path of $\alpha$-vertices of length $(2r+1)k$ connecting $a$ and $b$. By choosing every $(2r+1$)rst vertex on the path, we find a set of $k$
$\alpha$-vertices of pairwise distance greater than $2r$, contradicting our
assumption.
\item There exist at most $\ell<k$ $\alpha$-components $C_1,\ldots, C_\ell$.
By definition of $\alpha$-edges, two $\alpha$-vertices of different
$\alpha$-components have distance greater than $2r$. Hence, we can
choose at least one $\alpha$-vertex from every $\alpha$-component
into a set of $\alpha$-vertices at pairwise distance greater than $2r$.
\item Every $\alpha$-vertex belongs to a unique $\alpha$-component
and two $\alpha$-vertices from different components have distance
greater than $2r$. Hence, a maximum set of $\alpha$-vertices of pairwise
distance greater than $2r$ in $G$ consists of maximum sets of
$\alpha$-vertices of pairwise distance greater than $2r$ in the $\alpha$-components, and the last claim follows.
\end{enumerate}

Conversely, assume that conditions 2 and 3 are satisfied. By the same
arguments as above, we conclude that $G$ does not contain a set of
$k$ $\alpha$-vertices with pairwise distance greater than
$2r$ in $G$.

Now compute for every $\alpha$-component $C_i$ the size $k_i$ of a
maximum set of $\alpha$-vertices from~$C_i$ with pairwise distance
greater than $2r$. We obtain the graph $H$ by adding to every $\alpha$-vertex in component $C_i$ the number $k_i$ (as a unary predicate $P_{k_i}$).

We are now ready to translate $\phi$ into a disjunction of
$\Sigma_{2,1}$-sentences $\psi_i'$
over $H$, as claimed.
Let $\psi_0'=\forall y \neg P_\alpha(y)$. For $1\leq i<k$, let

\begin{align*}
\psi_i'=\exists x_1\ldots\exists x_i \forall y\Big(\bigwedge_{1\leq j<j'\leq i} x_j\neq x_{j'}\wedge & \bigwedge_{1\leq j\leq i} P_\alpha(x_j)\wedge \\ \big(P_\alpha(y)\rightarrow
&  \bigvee_{1\leq j\leq i}
\dist_\alpha(y, x_j)< (2r+1)k\big) \wedge \sum_{1\leq j\leq i}k_i<k\Big).
\end{align*}
Here, $\dist_\alpha$ refers to the distance with respect to $\alpha$-edges,
which is definable, and \linebreak \mbox{$\sum_{1\leq j\leq i}k_i<k$} is an abbreviation for the
formula that consists of all disjunctions
of possible predicates representing numbers that sum up to a number smaller
than $k$.

Now $H\models \psi_i'$ if and only if there exist exactly $i$ $\alpha$-components
$C_j$ (represented by the vertex $x_j$, each of diameter (with respect to $\alpha$-edges) smaller than $(2r+1)k$ and such that $\sum_{1\leq j\leq i}k_i<k$. As proved above, $G$ does not satisfy $\phi$ if and only if $H'\models
\psi_i'$ for some $i<k$.
\end{proof}

By translating the formula $\phi$ into conjunctive normal form and then
eliminating negations by the reduction presented in \cref{claim:negation},
it suffices now to show how to translate disjunctions and conjunctions
of $\Sigma_{2,1}$-formulas again into a $\Sigma_{2,1}$-formula.

This translation is straightforward for conjunction. Let $G$ be any graph and
$\phi$ $\psi$ any two $\Sigma_{2,1}$ formula.
Renaming the variables if necessary, we have that
$\phi = \exists \x\ \forall z\ \phi'(\x,z)$ and
$\psi = \exists \y\ \forall z\ \psi'(\y,z)$, where $\psi'$ and $\phi'$ are quantifier free.
We then have that $G\models \phi \wedge \psi$ if and only if
$G\models \exists \x\ \exists \z\ \forall y (\phi'(\x,y) \wedge \psi'(\z,y))$

This translation is not as simple for disjunction. In addition, it requiers
the graph to have at least two vertices. If it not the case, the \Local model
trivially solves our issues.


\begin{claim}\label{cla:disjunctions}
  For any graph $G$ of size at least 2,
  for any two $\Sigma_{2,1}$ formulas:
  $\phi = \exists \x\ \forall z\ \phi'(\x,z)$ and
  $\psi = \exists \y\ \forall z\ \psi'(\y,z)$,
  we have that $G\models (\phi \vee \psi) \leftrightarrow \theta$, where
  \[ \theta := \ \exists \x\ \exists \y\ \exists w_1\ \exists w_2\ \forall z
    \bigg(\Big( (w_1 = w_2) \rightarrow \phi'(\x,z) \Big)
    \wedge \Big( (w_1 \neq w_2) \rightarrow \psi'(\y,z) \Big)\bigg).
  \]
\end{claim}
\begin{proof}
  Assume that $G\models \phi\vee\psi$, and assume first that $G\models \phi$.
  Let $\a$ be the elements of $G$ witnessing that, i.e.~$G\models \forall z \phi'(\a,z)$,
  then let $v$ be any vertex of $G$ and assigne $\y,w_1,w_2$ to this $v$.
  For all $c$ in $G$, we have $G\models \phi'(\a,c)$. We also have
  $G\models (v \neq v) \rightarrow \psi'(\ov,c)$.
  Therefore, $G\models \theta$.

  If we had that $G\models \psi$ only, we would need to assigne $\y$ to the witnessing
  elements, $\x,w_1$ to any element $v$ and $w_2$ to any element $u\neq v$. The conclusion remains the same.

  For the other direction, assume that $H\models \theta$ and let $\a,\b,c_1,c_2$ be the witnesses.
  If $c_1=c_2$, we derive that $G\models \forall z\phi'(\a,z)$ and therefore $G\models \phi$.
  Similarly, if $c_1\neq c_2$, we have $G\models \psi$. 
\end{proof}

\end{proof}

\begin{lemma}\label{lem:w1-w2}
  $\LocalW{1} \subsetneq \LocalW{2}$.
\end{lemma}

\begin{proof}
  We prove that the \textsc{Clique Domination} problem is not in $\LocalW{1}$.
  Assume towards a contradiction that there is a \Local-reduction from the \textsc{Clique Domination} problem to the \textsc{Multicolored Independent
  Set} problem. This is enough as the upcoming \cref{lem:indep_MC} shows that this problem is complete for \LocalW{1}.

  Let $(s,r,t,p)$ be the functions given by the reduction and $A$ be the corresponding \Local-algorithm.
  We fix $k=1$, and $\ell=3$, so the parameter is $3+1=4$.
  We define a graph $G$ that is a YES-instance for the $(k,\ell)$-\textsc{Clique Domination} problem as the graph composed of:
  \begin{enumerate}
    \item a node $v$,
    \item $p(4)+2$ triples of nodes $(a_i,b_i,c_i)$,
    \item $p(4)+2$ disjoint paths $P_i$, of length $t(4)$, connecting $v$ to each $a_i$,
    \item the edges $(a_i,b_i)$ and $(a_i,c_i)$ for every $i$, and
    \item the additional edge $b_1,c_1$.
  \end{enumerate}
  Intuitively, $G$ is a tree of depth $t(4)+1$ such that in the reduction $A$, no message can go from a leaf to the root. In the first leaf, $a_1,b_1,c_1$ form the only $3$-clique of the graph.
  Let $G'$ be the graph given when applying the reduction $A$ on $G$.
  Remember that we have a mapping~$\nu$ from the vertices of $G'$ to the vertices of $G$.
  The graph $G'$ must contain a multicolored independent set $I$ of size $p(4)$. We can therefore find an integer $j\le p(4+2)$
  such that $j>1$ and $\nu(I) \cap \{P_j,a_j,b_j,c_j\}=\emptyset$. Intuitively, this means that the $j$th branch of $G$ is not
  ``responsible'' for the creation of $I$.

  We now define $H$ as the copy of $G$ with only one extra edge: $(b_j,c_j)$. This time, $H$ is a no-instance for the
  $(1,3)$-\textsc{Clique Domination} problem. However, if running the reduction
  algorithm $A$, this yields a graph $H'$ that also contains a multicolored independent of size~$p(4)$.
  To see this, look at any vertex $x$ in $G$ such that $x\in \nu(I)$. By definition we have that $\dist(x,a_j)>t(4)$ both in
  $G$ and in $H$. So in $t(4)$ communication rounds, $x$ cannot distinguish between $G$ and $H$, and produces the same reduction.
  Therefore $I$ is also created in $H'$, which is a contradiction.
\end{proof}

\subsection{LOCAL-W[1] complet problems}



Since $\Sigma_{1,1}=\Sigma_1$, we have
$\LocalW{1} = \LocalA{1}$ (just as
in classical parameterized complexity theory). In the light
of the above collapse result it remains (at least for first-order
definable problems) to determine whether they belong to
$\LocalW{1}$. In the remainder of this section we
prove that the $\textsc{Multicolored Independent
Set}$ problem and the $\textsc{Induced Subgraph Isomorphism}$ problem
are complete for $\LocalW{1}$ under \Local reductions.
In classical parameterized complexity theory these problems
are prime examples of \textsf{W}$[1]$-complete problems,
however, the standard reductions do not translate to \Local
reductions, and we have to come up with new reductions.

We immediately have that $\textsc{Multicolored Independent}$
$\textsc{Set}$ $\leq_\Local$ $\textsc{Induced}$ \linebreak
$\textsc{Subgraph Isomorphism}$ $\leq_\Local$ \mbox{$\textsc{MC}$-$\Sigma_{1}$},
 as each one is a specific case of the following one.
We therefore only prove that:
\begin{lemma}\label{lem:indep_MC}
  $\emph{\textsc{MC}}$-$\Sigma_{1} \leq_\Local \emph{\textsc{Multicolored Independent}}$ $\emph{\textsc{Set}}$.
\end{lemma}
The key part of the proof of this lemma is to deal with disjunctions.
More precisely, fix a~$\Sigma_{1}$-formula expressing that there is a
red-blue independent set or a green-yellow one. In a classical parameterized
reduction we could just create two copies of the graph. In one copy we would
``un-color'' all nodes but those in red or blue. In the second one we would
un-color all but those in green or yellow. Then, we would draw a complete
bipartite graph between the two copies. There is a multicolored independent
set of size two in this new structure if and only if the original
$\Sigma_{1}$ formula was satisfied.

The issue with this approach is that this is not a \Local reduction due the
the bipartite connections between the copies, making the radius of the
reduction arbitrarily large. Hence, we have to come up with a new way of
dealing with disjunctions.

\begin{lemma}\label{lem:disjunctions}
Let $\PPP_1$ and $\PPP_2$ be parameterized problems with
$\PPP_1\leq_\Local \textsc{Multicolored} \linebreak \textsc{Independent Set}$ and
$\PPP_2\leq_\Local \emph{\textsc{Multicolored Independent Set}}$.\newline
Then also
$\PPP_1\cup\PPP_2 \leq_\Local \emph{\textsc{Multicolored Independent Set}}$.

\end{lemma}

\begin{proof}[Proof of \cref{lem:disjunctions}]
  Let $\PPP_1$ and $\PPP_2$ be two problems satisfying the requirements
  of the Lemma. Let $(s_1,r_1,c_1,t_1,p_1)$ (resp. 
  $(s_2,r_2,c_2,t_2,p_2)$) be the functions attesting that $\PPP_1$
  (resp.~$\PPP_2$) reduces to the $\textsc{Multicolored Independent Set}$
   problem. We define $s \coloneqq s_1+s_2$, $r\coloneqq \max(r_1,r_2)$, $t\coloneqq \max(t_1,t_2)$,
   and $p\coloneqq  p_1p_2$, and show that there is a \Local reduction that is
   $(s,r,t,p)$ bounded from $\PPP_1\mathbin{\cup} \PPP_2$ to the
   $\textsc{Multicolored Independent}$ $\textsc{Set}$ problem.

  Let $(G,k)$ be an instance of $\PPP_1\cup \PPP_2$. We compute $G'$ as follows. First, compute $(G_1,\nu_1)$ and $(G_2,\nu_2)$ given by the reductions from $\PPP_i$ to the $\textsc{Multicolored Independent Set}$ problem (for $i\in\{1,2\}$). We have that $G_1$ contains at most $p_1(k)$ different colors and name them $c_{1},\ldots,c_{p_1(k)}$, and $G_2$ contains at most $p_2(k)$ different colors and name them $d_{1},\ldots,d_{p_2(k)}$.

   We then introduce $p_1(k)p_2(k)$ new colors named $(e_{i,j})$  for $1\le i \le p_1(k)$ and $1\le j \le p_2(k)$. We then change $G_1$ and~$G_2$ in the following way. For every node $u$ in $G_1$ of color~$c_i$, we replace it with $p_2(k)$ new nodes, each of them with a distinct color among $\{e_{i,1},\ldots e_{i,p_2(k)} \}$.
   Similarly, for every node $v$ in $G_2$ of color $c_j$, we replace it with $p_1(k)$ new nodes, each of them with a distinct color among $\{e_{1,j},\ldots e_{p_1(k),j} \}$. Additionally, for every edge $\{u,v\}$ in either $G_1$ or $G_2$, we draw a complete bipartite graph between the nodes created by these two blow ups.

   We now claim that the obtained graph contains a multicolored independent set of size~$p(k)$ if and only if $G_1$ contains a multicolored independent set of size $p_1(k)$ or $G_2$ contains a multicolored independent set of size $p_2(k)$, and therefore if and only if $(G,k)$ belongs to $\PPP_1\cup \PPP_2$. It should be clear that if $G_i$ contains a multicolored independent set of size $p_i(k)$, then the obtained graph contains one of size $p(k)$.

   Assume now that this final graph contains a multicolored independent set of size $p(k)$. If there is an $i\le p_1(k)$ and a \mbox{$j\le p_2(k)$} such that this multicolored independent set contains no node obtained by the blow up of a node of color $c_i$ nor by the blow up of a node of color~$d_j$, then this set does not contain a node of color $e_{i,j}$.
   Since this set must contain a node of each color, either for all $i\le p_1(k)$ the set contains a node obtained by the blow up of a node in $G_1$ of color $c_i$, which yields a multicolored independent set of size $p_1(k)$ in $G_1$, or for all $j\le p_2(k)$ the set contains a node obtained by the blow up of a node in $G_2$ of color $d_j$, which yields a multicolored independent set of size $p_2(k)$ in $G_2$.

   The only missing part to make this an $(s,r,t,p)$-bounded \Local reduction is that the obtained graph is not connected yet. To do so, we assume that the reduction from $\PPP_1$ (resp.~$\PPP_2$) to the $\textsc{Multicolored Independent Set}$ problem has the extra property that for all $v$ in $G$, there is a node $u$ in $G_1$ (resp. $G_2$) with $\nu_1(u)=v$ (resp. $\nu_2(u)=v$). If that is the case, then the final step of our \Local reduction is to add, for every $v$ in $G$, an edge between any two nodes $u_1,u_2$, with $\nu_1(u_1)=\nu_2(u_2)=v$.

   This additional property can be enforced for a \Local reduction to the $\textsc{Multicolored}$ $\textsc{Independent Set}$. Indeed, an instance of this problem is not impacted by creating new uncolored nodes and adding edges between two nodes if one of them is uncolored. In this case, we can slightly change the \Local reduction from $\PPP_i$ to the $\textsc{Multicolored Indepen}$-
   $\textsc{dent Set}$ problem. After $G_i$ has been computed, for every node $v$ in $G$, we create one new uncolored node $v_i$. We then add edges connecting $v_i$ to all nodes $u$ in $G_i$ such that $\nu_i(u)=v$ and connecting $v_i$ to $u_i$ for every edge $(u,v)$ in~$G$.
 \end{proof}

The last ingredient that we need is the following lemma.
\begin{lemma}\label{lem:indep_sub-graph-iso}
  ~\newline$\emph{\textsc{Induced Subgraph Isomorphism} }
  \leq_\Local$  $\emph{\textsc{Multicolored Independent Set}}$.
\end{lemma}

\begin{proof}[Proof of \cref{lem:indep_sub-graph-iso}]
  Let $(G,H)$ be an instance of the $\textsc{Induced Subgraph}$
  $\textsc{Isomorphism}$ problem. The task is to determine whether
  $G$ contains an induced subgraph isomorphic to $H$. The parameter is
  $|H|$. We construct an instance $(G',k')$ of the
  $\textsc{Multicolored}$ $\textsc{Independent Set}$ problem as follows:

  Let $k'$ be the number of connected components of $H$. For each
  connected component $C$ of $H$, and for each induced subgraph $B$
  of $G$ isomorphic to $C$, we create a new node $v_{B,C}$. The vertices
  of $G'$ are all vertices $v_{B,C}$. We create an edge between $v_{B,C}$ and
  $v_{B',C'}$ if and only if $B$ intersects $B'$ or if there is an edge in $G$
  between a vertex of $B$ and a vertex of~$B'$.
  We use~$k'$ colors to color $G'$: each vertex $v_{B,C}$ gets color ``$C$''.

  The existence of an induced subgraph of $G$ isomorphic to~$H$ is
  equivalent to the existence of a multicolored independent set of
  size $k'$ in $G'$. We only have to check that this is indeed a \Local
  reduction.

  We define $\nu \colon V(G')\to V(G)$ as the mapping that
  assigns $v_{B,C}$ to the vertex of $B$ with the smallest identifier.
  By running a \Local algorithm for $|H|$ rounds, every vertex $u\in G$
  can detect whether it is the vertex with smallest identifier of a subgraph
  $B$ that is isomorphic to a component $C$ of $H$.
  Note also that for any edge $\{v_{B,C}, v_{B',C'}\}$ of
  $G'$, we have that $B'$ is included in the $2|H|$ neighborhood of
  $\nu(v_{B,C})$. Therefore, we can compute
  by a \Local algorithm in $2|H|$ rounds also all edges of $G'$.

  To summarize, we have that $|G'|\le |G|^{|H|}$, the radius of the
  reduction and the number of rounds are $2|H|$, and the new
  parameter $k'$ is smaller then $|H|$, making this a
  \Local reduction. However the congestion is unbounded, so
  this is not a \Congest reduction.
\end{proof}

With \Cref{lem:disjunctions} and \Cref{lem:indep_sub-graph-iso} proved, \Cref{lem:indep_MC} follows quite easily.

\begin{proof}[Proof of \Cref{lem:indep_MC}]
All $\Sigma_{1}$-formulas can be expressed as a disjunction of conjunctive queries with possibly negated atoms. The model checking of such queries are special cases of the $\textsc{Induced Subgraph Isomorphism}$ problem. By \Cref{lem:indep_sub-graph-iso}, this reduces to the $\textsc{Multicolored}$
 $\textsc{Independent Set}$ problem. The conclusion is a straightforward induction on the size of the disjunction, each step being solved by \Cref{lem:disjunctions}.
\end{proof}


\section{Kernelization}\label{sec:kernel}

We now turn our attention to \emph{distributed kernelization}.
Kernelization is a classical approach in parameterized complexity
theory to reduce the size of the input instance in a polynomial time
preprocessing step.
It is a classical result
of parameterized complexity that a problem is fixed-parameter
tractable if and only if it admits a kernel.
We give two definitions of distributed kernelization and study
their relation to fixed-parameter tractability.

\begin{definition}
A \Distributed kernelization algorithm is a
\Distributed algorithm that on input $(G,k)$ of a parameterized
problem $\PPP$ computes in $f(k)$ rounds an
equivalent instance $(G',k')$ of order at most $g(k)$ for
computable functions~$f,g$.
\end{definition}

Here, the graph $G'$ is represented in $G$ as in a
\Distributed reduction.
Obviously, if a
problem admits a \Distributed kernelization, then it lies
in \Distributed-\FPT. The converse however, depends on the
model of computation.



\begin{example}
The problem $D_d$ whether a graph contains a vertex of degree greater than
$d>3$ does not admit a \Local (and hence also not
a \Congest) kernelization parameterized by $d$.
\end{example}
\begin{proof}
Assume that there exists
a \Distributed kernelization algorithm $\mathcal{A}$
for $D_d$ which on input $(G,d)$
computes in~$f(d)$ rounds an equivalent instance $(G',d')$ of order at
most~$g(d)$. We consider the value $d=4$ and the following family of
graphs. The graph $G_n^0$ is a path on $n$ vertices $\{v_1,\ldots, v_n\}$,
where additionally we attach vertices $x_1, x_2$ to~$v_1$
and $y_1,y_2$ to $v_n$.
The graph $G_n^1$ is like$G_n^0$ but we additionally
attach one more vertex $x_3$ to $v_1$, $G_n^2$ is like $G_n^0$ but we additionally attach one more
vertex $y_3$ to $v_n$, and $G_n^3$ is like $G_n^0$ but we additionally
at both ends we attach~$x_3$ and~$y_3$. The ids of vertices from
$G_n^0,G_n^1,G_n^2$ are equal to those of $G_n^3$ restricted to
the respective domain. The instances $G_n^0$ are negative instances
of $D_4$, while $G_n^1, G_n^2$ and $G_n^3$ are positive instances.

Now consider the execution of $\mathcal{A}$ on $G_n^3$ for large $n$.
It produces an equivalent instance $(H_n^3, c)$ for some $c\in \N$
of order at most $g(4)$, which is represented by mappings $\nu$ and $\eta$
in $G$. As $H_n^3$ is connected, $\nu(V(H_n^3))$ is a connected
subgraph of order at most $g(4)$ of $G_n^3$. In fact, we may assume
that $\nu(V(H_n^3))=v_i$ for some vertex $v_i\in V(G_n^3)$,
as $\mathcal{A}$ is a local algorithm that can collect all local
information in a single node. Observe that $H_n^3$ depends only on
the $k\coloneqq g(4)+f(4)$-neighborhood of $v_i$, as $\mathcal{A}$
can access only information
about~$g(4)$ vertices around $v_i$ that it collects in
$f(4)$ rounds. We distinguish two cases.

First case: $k<i<n-k$. We now consider the execution of $\mathcal{A}$ on
$G_n^0$. As the $k$-neighborhoods of $v_i$ are isomorphic (with ids)
in $G_n^0$ and $G_n^3$, the produced graph $H_n^0$ must be isomorphic
to $H_n^3$ and also be mapped to the vertex $v_i$, and the produced
parameter $c$ is equal to the parameter produced for $H_n^3$. But this
is a contradiction to the fact that $G_n^0$ is a negative instance of $D_4$.

Second case: $i\leq k$ (the case $i\geq n-k$ is analogous). We consider
the execution of $\mathcal{A}$ on $G_n^2$ (or on $G_n^1$ in
case $i\geq n-k$). Because the
$k$-neighborhoods of all vertices $v_j$ for $j\geq k$ are isomorphic (with ids)
in $G_n^2$ and $G_n^3$, and $\mathcal{A}$ on~$G_n^3$ mapped~$H_n^3$ to
$v_i$ with $i\leq k$, $\mathcal{A}$ on $G_n^2$ must also map the
positive instance~$H_n^2$ to some $v_j$ with $j\leq k$. But then we
consider the execution of $\mathcal{A}$ on $G_n^0$. With the same
argument as above, $\mathcal{A}$ must produce the same instance
$H_n^0\cong H_n^2$ and map it to the same vertex $v_j$ on
which $\mathcal{A}$ mapped the instance~$H_n^2$. However,
$G_n^0$ is a negative instance, a contradiction.
\end{proof}

On the other hand, in the \CCongest model the two notions of
kernelization and fixed-parameter tractability are
equivalent.

\begin{lemma}
If $\PPP\in$ \CCongest-\textsf{FPT},
then $\PPP$ admits a \CCongest kernelization.
\end{lemma}
\begin{proof}
As $\PPP\in$ \CCongest-\textsf{FPT}, we may on an instance
$(G,k)$ the algorithm witnessing this. Now in the \CCongest
model, all nodes can broadcast their answer so that in the
next round all nodes know whether $(G,k)$ is a positive or a
negative answer. Now the kernelization algorithm can map a
hardcoded equivalent instance of constant size to the node with
minimum id.
\end{proof}

We give a second definition of a \emph{fully polynomial}
\Distributed kernelization algorithm,
which better reflects the intuition that kernelization
should express efficient preprocessing.

\begin{definition}
A fully polynomial \Distributed kernelization algorithm is a
$\mathsf{DISTRI}$-$\mathsf{BUTED}$
kernelization algorithm where additionally
we restrict the computational power of each node to time
polynomial in the input size.
\end{definition}

A fully polynomial \Distributed kernelization algorithm
can be simulated by a sequential algorithm in polynomial time.
Hence, we obtain that the class of problems that admits a fully polynomial \Distributed kernelization algorithm is a subset
of sequential \textsf{FPT}. It is an
interesting question which problems in \textsf{FPT}
actually admit a fully polynomial \Distributed
kernelization algorithm. As we intend to make a conceptual
rather than a technical contribution, we leave this
investigation for future work.

\section{XPL and model-checking on bounded expansion classes}\label{sec:mc-be}

Finally, we want to introduce a distributed analogue of the
parameterized  complexity class~$\mathsf{XP}$ of
\emph{slicewise polynomial} problems. This class
contains all problems that can be solved in time~$n^{g(k)}$ for some computable function $g$. This definition
obviously has to be adapted to make sense in the distributed setting,
as every problem can be solved in a polynomial number of
rounds (polynomial in the graph size) in the
\Congest~model. We define the following class $\Distributed$-$\mathsf{XPL}$, where
$\mathsf{XPL}$ stands for \emph{slicewise poly-logarithmic}.

\begin{definition}
The class $\Distributed$-$\mathsf{XPL}$ is the class of
problems that can be solved by a
\Distributed algorithm in $f(k)\cdot (\log n)^{g(k)}$ rounds
for computable functions $f$ and~$g$.
\end{definition}

The first-order model-checking problem belongs to the
sequential class \textsf{XP}. We can simply instantiate
the quantifiers of a formula $\phi$ in all possible ways and thereby
evaluate in time~$n^{\Oof(|\phi|)}$ whether $\phi$ is true in
the input graph $G$.
Since the question whether a graph contains
two blue nodes (a simple first-order property)
cannot be decided by a \Local algorithm in
a sublinear (in the diameter) number of rounds in general,
the problem does not lie in \Local-\textsf{XPL}.
We can also say that it is unlikely for the model-checking problem to belong
in \CCongest-\textsf{XPL},
as finding triangle in a poly-logarithmic number of rounds would wildly improved
the best known algorithm of $O(n^{2/3}(\log n)^{2/3})$~\cite{DBLP:conf/podc/IzumiG17}.

We therefore turn our attention to solve
the problem on restricted graph classes.
Two prominent
graph classes on which first-order model-checking is even
fixed-parameter tractable by sequential algorithms are
classes of bounded expansion~\cite{dvovrak2013testing}
and nowhere dense classes of graphs~\cite{grohe2017deciding}.

Very briefly, a graph~$H$ is a {\em{depth-$r$ minor}} of a
graph~$G$ if~$H$ can be obtained from a subgraph of~$G$
by contracting mutually disjoint connected subgraphs of radius
at most $r$. A class of graphs~$\CCC$ has {\em{bounded expansion}}
if there is a function $f\colon \N\to \N$ such that for every $r\in \N$,
in every depth-$r$ minor of a graph from $\CCC$ the ratio between
the number of edges and the number of vertices is bounded by
$f(r)$. More generally, $\CCC$ is {\em{nowhere dense}} if there is
a function $t\colon \N\to \N$ such that no graph from $\CCC$
admits the clique $K_{t(r)}$ as a depth-$r$ minor.
Every class of bounded expansion is nowhere dense, but the
converse does not necessarily hold~\cite{Sparsity}.
Class~$\CCC$ has {\em{effectively bounded expansion}}, respectively
is {\em{effectively nowhere dense}}, if the respective
function~$f$ or $t$ as above is computable.
Many classes of sparse graphs studied in the literature have
(effectively) bounded expansion, including planar graphs,
graphs of bounded maximum degree, graphs of bounded treewidth,
and more generally, graphs excluding a fixed (topological) minor.
A notable negative example is that classes with bounded {\em{degeneracy}}, equivalently with bounded {\em{arboricity}},
do not necessarily have bounded expansion, as there we have only a finite bound on the edge density in subgraphs (aka depth-$0$ minors).
We refer to the textbook~\cite{Sparsity} for extensive background
on the theory of bounded expansion and nowhere dense graph classes.

The methods used to establish
fixed-parameter tractability of the model-checking problem
on these classes do not yield distributed fixed-parameter
tractability. However, the model-checking result on
bounded expansion classes
has been reproved multiple times~\cite{GajarskyKNMPST18,grohe2011methods,
kazana2013enumeration,pilipczuk2018parameterized} with different methods.
We show how to combine these methods with methods
for distributed computing
from~\cite{nevsetvril2016distributed} and prove that
first-order model-checking on bounded expansion classes
lies in the class \CCongest-\textsf{XPL}.

\begin{theorem}\label{thm:mc}
Let $\CCC$ be a graph class of effectively bounded expansion.
Then there exists a computable function $f$ and a \CCongest algorithm
that given a vertex and edge colored graph $G\in\CCC$
and a first-order sentence $\phi$ decides in $f(|\phi|)\cdot
\log n$ rounds whether $\phi$ holds in $G$.
\end{theorem}

\Cref{thm:mc} states that the first-order model-checking problem on
classes of effectively bounded expansion belongs to
\CCongest-\textsf{XPL}.
Our proof of the theorem follows closely the lines of the proof
given in \cite{pilipczuk2018parameterized} and
we point out only where the proof has to be changed. The idea
of the proof is as follows. We first compute a so-called
\emph{low treedepth coloring} of the input graph,
and then use this
coloring to apply a quantifier elimination procedure
for first-order logic.
It is known that such colorings exist for graphs from classes of bounded
expansion~\cite{NesetrilM08} and furthermore that they
can be computed efficiently even in the \Congest
model~\cite{nevsetvril2016distributed}. For establishing
\Cref{thm:mc} it remains to revisit the quantifier elimination
procedure and show that it can be implemented in the
\CCongest model.
Let us now introduce the relevant definitions.

\begin{definition}
A \emph{rooted forest} is an acyclic graph $F$ together with
a unary predicate $R\subseteq V(F)$ selecting one root in each
connected component of $F$. A tree is a
connected forest.  The {\em{depth}} of a node $x$ in a rooted forest
$F$ is the distance between~$x$ and the root in the connected
component of $x$ in $F$. The depth of a forest is the largest depth of
any of its nodes.  The {\em{least common ancestor}} of nodes $x$ and
$y$ in a rooted tree is the common ancestor of $x$ and $y$ that has
the largest depth.
\end{definition}

\begin{definition}
An \emph{elimination forest} of a graph $G$ is a rooted forest $F$ on
the same vertex set as $G$ such that whenever $uv$ is an edge
in $G$, then either $u$ is an ancestor of $v$, or~$v$ is an
ancestor of $u$ in $F$. The {\em{treedepth}} of a graph $G$
is the smallest possible depth of a separation forest of $G$.
\end{definition}

For the sake of quantifier elimination it will be convenient to
encode rooted forests by a unary function
$\mathsf{parent}\colon V(F)\rightarrow V(F)$.
The function encodes a tree in the
expected way, every vertex is mapped to its parent in the
tree, while the root vertex is mapped to itself. In the following
we assume that trees are encoded via the $\mathsf{parent}$
function.

\begin{definition}
For an integer~$p$, a coloring $\lambda\colon V(G)\to
\{1,\ldots,M\}$ of a graph $G$ is
a {\em{$p$-treedepth coloring}} of $G$
if every $i$-tuple of color classes in $\lambda$, $i\leq p$, induces
in $G$ a graph of treedepth at most $i$.
\end{definition}

\begin{lemma}[\cite{NesetrilM08}]
A class $\CCC$ of graphs has bounded expansion
if and only if for every $p$ there is a number $M$ such that every
graph $G\in\CCC$ admits a $p$-treedepth coloring using $M$ colors.
\end{lemma}

In fact, we must work with a related notion, as for our
application we need to be able to compute the elimination
forests $F_I$ witnessing that an $i$-tuple $I$ of color classes
has treedepth at most $i$. While in the sequential setting we
can simply perform a depth-first search to compute an approximation
of such an elimination forest, it is unclear how to compute
such forests in \CCongest-\textsf{XPL}.

\begin{definition}
For an integer $p$, a $(p+1)$-centered coloring of a graph
$G$ is a coloring $\lambda:V(G)\rightarrow \{1,\ldots, M\}$
so that for any induced connected subgraph $H\subseteq G$,
either some color appears exactly once in $H$, or $H$ gets at
least $p+1$ colors.
\end{definition}

Every $(p+1)$-centered coloring is a $p$-treedepth coloring.
More precisely, we have the following lemma.

\begin{lemma}[Lemma 4.5 of~\cite{nevsetvril2006tree}]
Let $G$ be a graph and let $\lambda$ be a $(p+1)$-centered
coloring of $G$. Then any subgraph $H$ of $G$ of treedepth
$i\leq p$ gets at least $i$ colors in $\lambda$.
\end{lemma}

Furthermore, from a $(p+1)$-centered coloring with $M$ colors
one easily computes a forest~$F$ of height at most $i$, $1\leq i\leq p$,
for each tuple of at most $i$ color classes.

\begin{lemma}\label{lem:cc-td}
Given a graph $G$ and a $(p+1)$-centered coloring
$\lambda:V(G)\rightarrow\{1,\ldots, M\}$, we can compute
in $\Oof(p\cdot 2^p)$ rounds in the \Congest
model for every $i$-tuple $I$ of colors, $1\leq i\leq p$, an
elimination forest $F_I$ of height at most $i$.
\end{lemma}

\begin{proof}
We iterate through all $i$-tuples of color classes, $1\leq i\leq p$,
which leads to the factor $\Oof(M^p)$ in the above estimation on
the number of rounds in the algorithm. For each $i$-tuple~$I$ of colors, we can then compute an elimination forest
$F_I$ of height at most $i$ as follows. It is folklore
(see e.g.~Section 6.2 in~\cite{Sparsity}) that the longest path
in a graph of treedepth $i$ has length (number of edges) at
most $2^i-2$. We can hence compute the components of
$G[I]$ (the subgraph induced by the colors in $I$)
in $\Oof(2^i)$ rounds, by performing a breadth-first
search from every vertex, and whenever the searches from
two vertices meet, we continue only the search of the vertex
with the smaller id to avoid large congestion. Now each
component $C$ of $G[I]$ is connected and gets at most
$p$ colors, hence there is a vertex of unique color. We can
find such a vertex $v$ in $\Oof(2^i)$ rounds by traversing the
constructed bfs tree and keeping track of the encountered
colors. We now make $v$ the root of $F_I$ and recursively
continue to construct $F_I$ by decomposing the components
of $G[I]-\{v\}$ (which has one less color) as above.
After $i$ recursive steps, the procedure stops and
produces an elimination forest $F_I$ of depth at most $i$.
Observe that this construction is only possible in the
\CCongest model.
\end{proof}

We now appeal to the result of Ne\v{s}et\v{r}il and
Ossona de Mendez~\cite{nevsetvril2016distributed} that
$(p+1)$-centered colorings are computable in
\Congest-\textsf{XPL}.

\begin{lemma}[\cite{nevsetvril2016distributed}]\label{lem:ltd}
Let $\CCC$ be a class of graphs of effectively
bounded expansion. There exists a computable function $g$
and a \Congest algorithm that on input $G\in\CCC$ and
$p\in\N$ computes a $(p+1)$-centered coloring of $G$ with
$\Oof(1)$ colors in $g(p)\cdot \log n$ rounds.
\end{lemma}

We now come to the quantifier elimination procedure
on classes of bounded expansion.
The proof boils down to proving how to eliminate
a single existential quantifier for bounded depth forests.
This elimination is then lifted to bounded expansion classes
via low-treedepth colorings. The following statement is
an adapted version of Lemma 26
of~\cite{pilipczuk2018parameterized}, which is the crucial
ingredient of the proof.

\begin{lemma}[Lemma 26
of~\cite{pilipczuk2018parameterized} (adapted)]\label{lem:qe-trees}
Let $d\in \N$ and $\Lambda$ be a label set.
Then for every formula $\varphi(\bar x)\in \FO[\{\parent\}
\cup \Lambda]$ with $|\bar x|\geq 1$ and of the form
$\varphi(\bar x)=\exists y\, \psi(\bar x,y)$
where~$\psi$ is quantifier-free, and every $\Lambda$-labeled
forest $F$ of depth at most $d$, there exists a label set
$\wLambda$, a quantifier-free formula $\wphi(\bar x)\in
\FO[\{\parent\}\cup \wLambda]$, and a $\wLambda$-relabeling
$\wh{F}$ of~$F$ such that $\varphi$ on $F$ is
equivalent to $\wphi$ on $\wh{F}$.
Moreover, the label set $\wLambda$ is computable from~$d$
and~$\Lambda$, the formula~$\wphi$ is computable from
$\varphi,d,\Lambda$, and the transformation
which computes~$\wh{F}$ given $F$ can be done in
$f(d,|\phi|)$ rounds by a \CCongest algorithm for a
computable function $f$.
\end{lemma}
\begin{proof}[sketch]
We can follow the lines of the proof of Lemma 26
of~\cite{pilipczuk2018parameterized} and observe
that the new labels can be computed bottom up
along the tree by a \CCongest algorithm. For this
is suffices to count the number of types of the
descendants of a node up to a certain threshold.
Hence, the amount of information that has to be sent
and stored depends functionally only on $d$ and
$\phi$ and can be sent with low congestion along
the forest edges. In the case $h=0$ in the proof, we
crucially use that vertices from different subtrees of
the forest can communicate via communication
edges that are not edges of the forest.
\end{proof}

The rest of the proof works exactly as the proof
given in~\cite{pilipczuk2018parameterized} by replacing
all subroutines for computing low tree\-depth colorings and
elimination forests by \Cref{lem:cc-td} and \Cref{lem:ltd}.



\section{Conclusion}

In this work we followed the approach of
parameterized complexity to provide a framework
of parameterized distributed complexity.
We could only initiate the study of distributed parameterized
complexity classes and many interesting questions remain open.
On the one hand, the parameterized distributed complexity and distributed
kernelization complexity of many important
graph problems has not yet been studied. On the other hand,
it remains an interesting question to find parameterized
distributed reductions between commonly studied
graph problems.

\addcontentsline{toc}{section}{References}
\bibliographystyle{plain}
\bibliography{13_ref}


\end{document}